\documentclass{ifacconf}

\usepackage{color}   
\usepackage{graphicx}      
\usepackage{natbib}        

\usepackage{amsfonts}
\let\theoremstyle\relax 
\usepackage{amsthm}  
\usepackage{amsmath} 
\usepackage{amssymb}  
\usepackage{enumerate}
\usepackage{url}

\usepackage[caption=false]{subfig}

\usepackage{enumitem} 
\usepackage{bbm} 

\newcommand{\Null}{\mathrm{Null}}
\newcommand{\Range}{\mathrm{Range}}
\newcommand{\one}{\mathbf{1}}
\newcommand{\rank}{\mathrm{rank}}
\newcommand{\myspan}{\mathrm{span}}
\newcommand{\mydiag}{\mathrm{diag}}

\newcommand{\blkdiag}{\mathrm{blkdiag}}

\newcommand{\T}{\mathrm{T}}

\newcommand{\R}{\mathbb{R}}
\newcommand{\G}{\mathcal{G}}
\newcommand{\E}{\mathcal{E}}
\newcommand{\V}{\mathcal{V}}
\newcommand{\N}{\mathcal{N}}

\newcommand{\RB}{R_{\text{\scriptsize{B}}}}
\newcommand{\FB}{F_{\text{\scriptsize{B}}}}
\newcommand{\LB}{L_{\text{\scriptsize{B}}}}
\renewcommand{\S}{\mathcal{S}}

\usepackage{tikz,tikz-3dplot}
	\usetikzlibrary{shapes,arrows}
	\tikzstyle{frame} = [draw, -latex]
	\tikzstyle{line} = [draw]
	\tikzstyle{line2} = [draw, dashdotted]
	\tikzstyle{line3} = [draw, dashed]
	\tikzstyle{line3UD} = [draw, dashed]
	\tikzstyle{place} = [circle, draw=black, fill=white, thick, inner sep=2pt, minimum size=1mm]
	\tikzstyle{place2} = [circle, draw=black, fill=black, thick, inner sep=2pt, minimum size=1mm]
	\tikzstyle{placeRed} = [circle, draw=red, fill=red, thick, inner sep=2pt, minimum size=1mm]
	\tikzstyle{vertex} = [circle, draw=black, fill=black, thick, inner sep=2pt, minimum size=1mm]
\usepackage{tkz-euclide}
\usetikzlibrary{arrows}
\tikzset{
  arrow/.pic={\path[tips,every arrow/.try,->,>=#1] (0,0) -- +(.1pt,0);},
  pics/arrow/.default={triangle 90}
}

\newtheorem{definition}{Definition}  
 
\newtheorem{lemma}{Lemma} 
\newtheorem{theorem}{Theorem} 
\newtheorem{proposition}{Proposition}
 
\newtheorem{conjecture}{Conjecture}

\begin{document}
\begin{frontmatter}

\title{Characterizing bearing equivalence in directed graphs \thanksref{footnoteinfo}} 

\thanks[footnoteinfo]{The work of Zhiyong Sun was partially supported by a starting grant from Eindhoven Artificial Intelligence Systems Institute (EAISI), Eindhoven, the Netherlands.   }

\author[First]{Zhiyong Sun} 
\author[Second]{Shiyu Zhao} 
\author[Fourth]{Daniel Zelazo} 

\address[First]{Department of Electrical Engineering, Eindhoven University of Technology (TU/e), Eindhoven, The Netherlands (e-mail: z.sun@tue.nl, sun.zhiyong.cn@gmail.com)}
\address[Second]{School of
Engineering, Westlake University, Hangzhou 310024, China (e-mail: zhaoshiyu@westlake.edu.cn).}
\address[Fourth]{Faculty of Aerospace Engineering, Technion-Israel Institute of Technology, Haifa, Israel. (e-mail: dzelazo@technion.ac.il).}

\begin{abstract}                
In this paper, we study bearing equivalence in directed graphs. We first give a strengthened definition of bearing equivalence based on the  \textit{kernel equivalence} relationship between bearing rigidity matrix and bearing Laplacian matrix. We then present several conditions to characterize bearing equivalence for both directed acyclic and cyclic graphs. These conditions involve the spectrum and null space of the associated bearing Laplacian matrix for a directed bearing formation. For directed acyclic graphs, all eigenvalues of the associated bearing Laplacian are real and nonnegative,   while for directed graphs containing cycles, the bearing Laplacian can have eigenvalues with negative real parts. Several examples of bearing equivalent and bearing non-equivalent formations are given to illustrate these conditions. 

\end{abstract}


\end{frontmatter}

\section{Introduction}

    
    

    

Recent years have witnessed a growing interest in bearing-based distributed control and estimation of networked multi-agent systems, such as bearing-based formation control (\cite{zhao2015bearing,  tron2016bearing, karimian2021bearing,   tang2022relaxed}) and network localization (\cite{lin2016distributed, arrigoni2018bearing}).  The application of bearing measurements in networked systems has been motivated by the advance of vision-based sensors, which  facilitate the sensing of relative directions between spatially distributed agents. This is in contrast to relative position or distance measurements, which tend to be more costly or unreliable. 

One of the key graph concepts underpinning these bearing-based applications is the bearing rigidity theory, which has been thoroughly discussed in  \cite{zhao2014TACBearing} with a focus on undirected graphs. The bearing rigidity theory has provided a powerful framework for studying bearing-based formation control of multi-agent systems, where the desired target shape is specified by constant inter-agent bearings \cite{zhao2019bearing, tang2021formation}. Compared to displacement-based (\cite{oh2015survey}) or distance-based  (\cite{sun2017distributed}) approaches, bearing-based formation control provides a natural solution to the problem of formation scale control where the shape of the formation is invariant but the inter-agent distance changes. 
One fundamental problem of bearing-based formation control that has not been completely solved is how to realize a desired bearing  formation   over \textit{directed} graphs. Most existing works on bearing-based formation or localization  assume that the underlying graph is undirected (\cite{zhao2019bearing}) or satisfies special directed graph structures 
(\cite{tang2021formation}).
Furthermore, the existing bearing-based formation control laws may also become unstable for directed graphs. 

As a key step towards solving the general problem of bearing-based formation control for directed graphs, the paper \cite{zhao2015bearing} first studied \textit{bearing persistence}   for bearing formations in directed graphs. The notion of bearing persistence follows in spirit the work of distance rigidity and persistence originally presented in  \cite{hendrickx2007directed, anderson2008rigid}. In \cite{zhao2015bearing}, however, the definition of  bearing-based persistence  involves the kernel equivalence between the bearing rigidity matrix (for undirected graphs) and the bearing Laplacian matrix (for directed graphs).   To differentiate between these notions, in this work we term the condition as \textit{bearing equivalence}. However, different to the full development of distance  persistence and its application in distance-based multi-agent control, conditions for bearing equivalence are not well understood and its characterization still  remains open.   
In this paper, following \cite{zhao2015bearing}, we aim to provide several conditions to characterize bearing equivalence in   directed graphs. These conditions involve the spectrum and null space of the associated bearing Laplacian matrix, which are separately presented for bearing-based directed graphs with or without directed cycles. 

The remaining parts of this paper are organized as follows. In Section~\ref{sec:prelim} we review key concepts and preliminaries on graph theory, bearing rigidity and useful matrix results. Section~\ref{Sec:BP_notion} introduces the notion of bearing equivalence and derives a useful formula for bearing Laplacian. Section~\ref{Sec:BP_acyclic} and Section~\ref{Sec:BP_cyclic} present several key conditions to characterize bearing equivalence for directed graphs without cycles and with cycles, respectively. Section~\ref{Sec:conclusion} concludes this paper. 
 
\section{Preliminaries} \label{sec:prelim}
\subsection{Notations}
The notations in this paper are fairly standard. 
We use  $ \mydiag(A_i)\triangleq\blkdiag\{A_1,\dots,A_n\}\in\mathbb{R}^{np\times nq}$ to denote block diagonal matrix with given $A_i\in\mathbb{R}^{p\times q}$ for $i=1,\dots,n$.
Let $\Null(\cdot)$ and $\Range(\cdot)$ be the null space and range space of a matrix, respectively, and $\dim(\cdot)$ be the dimension of a linear space.
Denote $I_d\in\R^{d\times d}$ as the identity matrix, and $\one_n \triangleq[1\,\cdots\,1]^\T$ (when the subscript of $\one$ is omitted, the vector dimension should be clear from the context).
Let $\|\cdot\|$ be the Euclidean norm of a vector or the spectral norm of a matrix, and $\otimes$ be the Kronecker product. For any nonzero vector $x\in\R^d$ ($d\ge2$), we define the projection matrix $P: \R^d\rightarrow\R^{d\times d}$ as
\begin{align} \label{eq:projection_matrix}
    P(x) \triangleq I_d - \frac{x}{\|x\|}\frac{x^\T }{\|x\|}.
\end{align}
For notational simplicity, we also denote $P_x=P(x)$.

\subsection{Preliminaries on graph theory}
Consider a directed graph with $n$ vertices and  $m$ edges, denoted by $\mathcal{G} =( \mathcal{V}, \mathcal{E})$.  The vertex set $\mathcal{V} = \{1,2,\ldots, n\}$ represents the index of $n$ agents in the group, and the edge set $\mathcal{E} \subset \mathcal{V} \times \mathcal{V}$ indicates the interconnection or neighboring relationship of $n$ agents. If $(i,j)\in\E$, we say agent $i$ can ``see'' agent $j$, which means agent $i$ can access the relative information of agent $j$. For agent $i$, its  neighbor set $\mathcal{N}_i$   is defined as $\mathcal{N}_i: = \{j \in \mathcal{V}: (i,j) \in \mathcal{E}\}$.
The incidence matrix $H = \{h_{ij}\} \in \mathbb{R}^{m \times n}$ for a directed graph $\mathcal{G}$ is defined by  $h_{ki} = 1$ if the  $k$th edge sinks at node $i$, or $h_{ki} = -1$ if the  $k$th edge  leaves  node $i$, or $h_{ki} = 0$ otherwise. Furthermore, its Laplacian matrix   $L(\mathcal{G}) = \{L_{ij}\} \in \mathbb{R}^{n \times n}$ is   defined as   $L_{ii} = |\mathcal{N}_i|$,  $L_{ij} = -1$ if $(i,j)\in\E$ and $L_{ij} = 0$ if $(i,j)\notin\E$. It is well known that  $\text{rank}(L(\mathcal{G})) = n-1$ with 0 being a simple eigenvalue of $L(\mathcal{G})$ if and only if the directed graph   $\mathcal{G}$ contains a (directed) spanning tree  (see e.g., \cite{mesbahi2010graph}).

\subsection{Preliminaries on bearing rigidity}
Bearing rigidity theory plays a key role in the analysis of bearing-based distributed formation control and network localization  problems.
In this section, we review key notions and results in the bearing rigidity theory presented in \cite{zhao2014TACBearing}. 

Given a finite collection of $n$ points $\{p_i\}_{i=1}^n$ in $\R^d$ ($n\ge2$, $d\ge2$), denote $p=[p_1^\T\,\cdots\,p_n^\T]^\T\in\mathbb{R}^{dn}$.
A \emph{formation} in $\R^d$, denoted as $\G(p)$, is a directed graph $\G=(\V,\E)$ together with $p$, where vertex $i\in\V$ in the graph is mapped to the point $p_i$.
For a formation $\G(p)$, define the \emph{edge vector} and the \emph{bearing}, respectively, as
\begin{align*}
e_{ij}\triangleq p_j-p_i, \quad g_{ij}\triangleq e_{ij}/\|e_{ij}\|, \quad \forall(i,j)\in\E.
\end{align*}
The bearing $g_{ij}$ is a unit vector.

\begin{definition}[Bearing Equivalent Formations]\label{definition_bearingEquivalence}
Two formations $\G(p)$ and $\G(p')$ are \emph{bearing equivalent} if $P_{g_{ij}}g'_{ij}=0$ for all $(i,j)\in\E$.
\end{definition}

By Definition~\ref{definition_bearingEquivalence}, bearing equivalent formations have parallel inter-neighbor bearings.
Suppose $|\E|=m$ and index all the directed edges from $1$ to $m$.
Re-express the edge vector and the bearing as $e_{k}$ and $g_{k}\triangleq {e_{k}}/{\|e_{k}\|}$, $\forall k\in\{1,\dots,m\}$.
Let $e=[e_1^\T \,\cdots\,e_m^\T ]^\T$ and $g=[g_1^\T \,\cdots\,g_m^\T ]^\T$.
Note $e$ satisfies $e=\bar{H}p$ where $\bar{H}=H\otimes I_d$ and $H$ is the incidence matrix of the graph $\mathcal{G}$.
Define the \emph{bearing function} $F_B: \R^{dn}\rightarrow\R^{dm}$ as
\begin{align*}
    \FB(p)\triangleq [g_1^\T, \cdots,g_m^\T]^\T.
\end{align*}
The bearing function describes all the bearings in the formation.
The \emph{bearing rigidity matrix} is defined as the Jacobian of the bearing function,
\begin{align}\label{eq_rigidityMatrixDefinition}
    \RB(p) \triangleq \frac{\partial \FB(p)}{\partial p}\in\R^{dm\times dn}.
\end{align}
Let $\delta p$ be a variation of $p$.
If $\RB(p)\delta p=0$, then $\delta p$ is called an \emph{infinitesimal bearing motion} of $\G(p)$.

\begin{definition}[{Infinitesimal Bearing Rigidity}]\label{definition_infinitesimalParallelRigid}
    A formation is \emph{infinitesimally bearing rigid} if all the infinitesimal bearing motions of the formation are trivial (i.e., translation and scaling of the entire formation).
\end{definition}

\begin{lemma}[\cite{zhao2014TACBearing}]\label{lemma_BearingRigidityProperty}
For any formation $\G(p)$, the bearing rigidity matrix defined in \eqref{eq_rigidityMatrixDefinition} satisfies
\begin{enumerate} 
\item $\RB(p)= \mydiag\left({P_{g_k}}/{\|e_k\|}\right)\bar{H}$;
\item $\rank(\RB)\le dn-d-1$ and $\myspan\{\one\otimes I_d, p\}\subseteq \Null(\RB)$.
\end{enumerate}
\end{lemma}


\begin{theorem}[\cite{zhao2014TACBearing}]\label{theorem_conditionInfiParaRigid}
    For any formation $\G(p)$, the following statements are equivalent:
    \begin{enumerate} 
    \item $\G(p)$ is {infinitesimally bearing rigid};
    \item $\G(p)$ can be uniquely determined up to a translation and a scaling factor by the inter-neighbor bearings;
    \item $\rank(\RB)=dn-d-1$;
    \item $\Null(\RB)=\myspan\{\one\otimes I_d, p\}$.
    \end{enumerate}
\end{theorem}

\subsection{Useful matrix results}
The following matrix results will be frequently used in this paper to characterize bearing rigidity/equivalence and stability of bearing-based formations. 
 
\begin{lemma} (Null space of matrix product, \cite{sun2017distributed}) \label{lemma_nullAB}
Consider two matrices $A \in \mathbb{R}^{m\times n}$ and $B \in \mathbb{R}^{n\times k}$ and the matrix product $C = AB$. Then there holds $\dim(\Null(C)) = \dim(\Null(B)) +  \dim(  (\Null(A)  \cap \Range(B)) $. In particular,  it holds  $\Null(C) = \Null(B)$ if and only if  $\Null(A)  \cap \Range(B)  = \{0\}$. 
\end{lemma}

\begin{lemma} \label{lemma:triangular_matrix} (Block  triangular   matrix, \cite{harville2008matrix})
Consider a real-valued block upper triangular square  matrix $A$, with the $i$-th diagonal square block denoted by $A_{ii}$. Then the eigenvalues of $A$ are the union of the set of eigenvalues of each diagonal  block $A_{ii}$; i.e., it holds $\lambda(A) = \bigcup_{i=1}^n \lambda (A_{ii})$. 
\end{lemma}
 
\begin{lemma}\label{lemma:projection} (Properties of projection matrices, \cite{bernstein2018scalar})
For a projection matrix $P_x$ defined in \eqref{eq:projection_matrix}, it holds $P_x^\T =P_x$, $P_x P_x=P_x$, and $P_x$ is positive semi-definite.
Moreover, $\Null(P_x)=\myspan\{x\}$ and the eigenvalues of $P_x$ are $\{0,1^{(d-1)}\}$.
Any two nonzero vectors $x,y\in\R^d$ are parallel if and only if $P_xy=0$ (or equivalently $P_yx=0$).
\end{lemma}

\begin{lemma} (Sum of projection matrices) \label{lemma:sum_projection}
Consider a set of projection matrix,  $P_i(x_i)$ defined in \eqref{eq:projection_matrix} associated  with nonzero vector $x_i\in\R^d$ ($d\ge2$). Then the following holds:
\begin{enumerate} 
    \item For the matrix sum $P_i(x_i) + P_j(x_j)$, if the vectors $x_i$ and $x_j$ are parallel (i.e., linearly dependent), then $\Null(P_i(x_i) + P_j(x_j))=\myspan\{x_i\}$ and $P_i(x_i) + P_j(x_j)$ is positive semi-definite. 
    \item Otherwise, if the vectors $x_i$ and $x_j$ are linearly independent, then $\Null(P_i(x_i) + P_j(x_j))=\{0\}$ and the matrix sum $P_i(x_i) + P_j(x_j)$ is positive definite. 
    \item For the matrix sum $P_i(x_i) + P_j(x_j) + \cdots + P_k(x_k)$, if \textbf{all} vectors $x_i, x_j, \cdots, x_k$ are parallel, then $$\Null(P_i(x_i) + P_j(x_j) + \cdots + P_k(x_k))=\myspan\{x_i\}.$$
    \item The matrix sum $P_i(x_i) + P_j(x_j) + \cdots + P_k(x_k)$ is positive definite if there exist at least two vectors in the list that are non-parallel. 
\end{enumerate}

 \end{lemma}
 The proofs of the above lemmas will be provided in the journal version of this paper. 

\section{Notion of bearing equivalence} \label{Sec:BP_notion}

The development of   bearing   equivalence  is motivated by the bearing-based formation control in directed graphs. The bearing-based formation control problem specifies a set of desired inter-agent bearings, $g_{ij}^*$, and the objective is to design a distributed control using only relative position measurements from neighboring agents to drive the agents to the formation specified by the $g_{ij}^*$'s, i.e., $\lim_{t\to\infty} g_{ij}(t) = g_{ij}^*$.  \cite{zhao2015bearing} proposed the following control scheme,
\begin{align}\label{eq_controlLaw}
\dot{p}_i(t)= - \sum_{j\in\N_i} P_{g_{ij}^*}\left(p_i(t)-p_j(t)\right), \quad i\in\V,
\end{align}
where $P_{g_{ij}^*}=I_d-(g_{ij}^*) (g_{ij}^*)^\T$. The stability and convergence of the above distributed and linear formation control system depends on the spectrum and null space of the \textbf{bearing Laplacian} $\LB\in\R^{dn\times dn}$, 
whose $ij$th block submatrix  is defined as below
\begin{align*}
\left\{
  \begin{array}{ll}
      [\LB]_{ij}=0_{d\times d}, & i\ne j, (i,j)\notin\E, \\
      {[\LB]_{ij}}=-P_{g_{ij}^*}, & i\ne j, (i,j)\in\E, \\ 
      {[\LB]_{ii}}=\sum_{j\in\N_i}P_{g_{ij}^*}, & i\in\V. \\
  \end{array}
\right.
\end{align*}
In this way, the compact matrix expression of the control law \eqref{eq_controlLaw} is
\begin{align}\label{eq_controlLaw_matrix}
\dot{p}(t)=-\LB p(t).
\end{align}
It becomes apparent that  the matrix $\LB$ can also be interpreted as a \textit{matrix-weighted} graph Laplacian with orthogonal projection matrix weights.   The bearing Laplacian $\LB$, together with its spectrum and null space, are jointly determined by the topological structure of the underlying graph and bearing information of the formation. 

For an undirected graph, it is easy to show that
\begin{align} \label{eq:LB_undirected}
    \LB= \bar{H}^\T  \mydiag{(P_{g_k})}\bar{H} 
\end{align}
which immediately leads to that, for   an undirected formation $\G(p)$, the bearing Laplacian $\LB$ is symmetric positive semi-definite and satisfies
$
\Null(\LB)=\Null(\RB)
$.

However, when the graph is directed, the bearing Laplacian is not symmetric and the above \textit{kernel equivalence} relationship would not hold in  general. The notion of   bearing equivalence is defined first in \cite{zhao2015bearing} (where it was termed `bearing persistence') based on the kernel equivalence condition as $\Null(\RB)=\Null(\LB)$.  In this paper, we shall introduce a strengthened definition of bearing equivalence as below.  
\begin{definition}[Bearing Equivalence]\label{definition_bearingequivalence}
A directed formation $\G(p)$ is \emph{bearing kernel equivalent} (in short, bearing equivalent) if $\Null(\RB)=\Null(\LB) = \myspan\{\one\otimes I_d, p\}$. 
\end{definition}
 
The problem of characterizing favourable properties of bearing Laplacian $\LB$ and bearing equivalence is motivated by the bearing-based formation system \eqref{eq_controlLaw} in directed graphs. In particular, 
\begin{itemize}
    \item The spectrum of $\LB$ determines the stability properties of the formation system \eqref{eq_controlLaw}. Preferably, we aim to find conditions to guarantee that $\LB$ has all eigenvalues with non-negative real parts such that the formation system  \eqref{eq_controlLaw} is stable. 
    \item The null space of $\LB$ determines the converged formation shape of the system \eqref{eq_controlLaw}.  Preferably, we aim to find conditions to guarantee  $\Null(\LB) = \myspan\{\one\otimes I_d, p\}$ such that the converged formation shape is bearing  equivalent to the target formation.
\end{itemize}

Graph conditions to guarantee bearing equivalence for a bearing formation still remain open, though some partial solutions were presented in \cite{zhao2015bearing}.  In this paper, we will revisit these conditions from \cite{zhao2015bearing} and present more (necessary and/or sufficient) conditions to characterize bearing equivalence. In particular, we will give several graph topological conditions to characterize  
the spectrum and null space of $\LB$, which underpin certain key requirements to ensure the formation convergence by the bearing control law \eqref{eq_controlLaw}.

\subsection{A useful formula for $\LB$}
As a counterpart of the matrix expression \eqref{eq:LB_undirected} of bearing Laplacian for undirected graphs, we first derive the following expression of bearing Laplacian for directed graphs. 
For a bearing formation with a directed graph $\mathcal{G}$, the associated bearing Laplacian can be expressed by 
\begin{align}\label{eq:formula_LB}
    \LB = \bar J^\T  \mydiag{(P_{g_k})}\bar{H}  
\end{align}
where $\bar{J}=J\otimes I_d$ and the matrix $J$ is obtained by replacing the `$-1$' entry of the incidence matrix $H$ by ` $0$' from    the directed graph. 
 
 
The bearing Laplacian formula of \eqref{eq:formula_LB} follows from the formula of the conventional Laplacian matrix $L$ for a directed graph: $L =  J^\T H$. By augmenting the Kronecker product and the matrix weight (in terms of the projection matrix $P_{g_k}$ associated to each directed edge), one can obtain  \eqref{eq:formula_LB}.

The formula of \eqref{eq:formula_LB} gives the following set inclusion of null spaces (Theorem 4 of \cite{zhao2015bearing}).
\begin{lemma}\label{result_nullLBNullRB}
For a directed formation $\G(p)$, the bearing Laplacian $\LB$ satisfies
\begin{align} \label{eq:null_space_LB_RB}
\myspan\{\one\otimes I_d, p\}\subseteq\Null(\RB)\subseteq \Null(\LB).
\end{align}
\end{lemma}
\begin{proof}
This lemma is proved in \cite{zhao2015bearing}, and here we give a shorter proof based on the formula~\eqref{eq:formula_LB}. First we define $\tilde{R}_B=\mydiag(\|e_k\|)\RB$ where $\RB$ is the bearing rigidity matrix. Then it holds $\Null(\tilde{R}_B)=\Null(\RB)$ and therefore $\myspan\{\one\otimes I_d, p\}\subseteq  \Null(\tilde{R}_B)$. Note that $\LB = \bar J^\T \tilde{R}_B$ and one has $\Null(\tilde{R}_B)\subseteq \Null(\LB)$, which concludes the set inclusion of \eqref{eq:null_space_LB_RB}. 
\end{proof}

We remark that, by the property of projection matrix in Lemma~\ref{lemma:projection}, an alternative   formula for $\LB$ is given as below
\begin{align}\label{eq:formula}
    \LB = \bar J^\T  \mydiag{(P_{g_k})}\bar{H} =  \underbrace{\bar{J}^\T \mydiag{(P_{g_k}^\T )}}_{\tilde{J}_B^\T }\underbrace{\mydiag{(P_{g_k})}\bar{H}}_{\tilde{R}_B}.
\end{align}
Based on this formula \eqref{eq:formula}, we give a necessary and sufficient condition to guarantee bearing equivalence. 
\begin{theorem}\label{theorem_NULLLB=1andp}
For a directed formation $\G(p)$, the equality $$\Null(\LB)=\myspan\{\one\otimes I_d, p\}$$ holds if and only if the following two conditions are both satisfied
\begin{enumerate} \renewcommand{\theenumi}{\Roman{enumi}}
    \item $\Null(\RB)=\myspan\{\one\otimes I_d, p\}$ (i.e., the formation is   infinitesimally bearing rigid), and 
    \item $\Null({\tilde{J}_B^\T}) \cap \Range ({\tilde{R}_B}) = \{0\}$. 
\end{enumerate}
\end{theorem}
\begin{proof}
This necessary and sufficient condition follows from the condition of infinitesimally bearing rigidity in  Theorem~\ref{theorem_conditionInfiParaRigid} and the null space property of $\LB = \tilde{J}_B^\T \tilde{R}_B$ as in \eqref{eq:formula} by applying Lemma~\ref{lemma_nullAB}. 
\end{proof}
It is not clear at this stage how to use the second condition, or what it means. In the following sections, we will derive more concrete conditions to characterize the spectrum and null space of $\LB$.


\section{Bearing equivalence in\\ directed acyclic graphs} \label{Sec:BP_acyclic}
In this section, we focus on directed acyclic graphs and present several conditions on bearing equivalence. 
\subsection{Spectrum of bearing Laplacian}
The first result characterizes the spectrum property of $\LB$ for directed acyclic graphs. 

\begin{proposition}
For directed acyclic graphs, the eigenvalues of the
bearing Laplacian $\LB$ are \textbf{real and  nonnegative}. 
\end{proposition}
\begin{proof}  
For directed acyclic graphs, the bearing Laplacian can always be reconstructed (with permutation of vertices and edges) as a block triangular matrix, while the $i$-th diagonal block consists of a projection matrix $P_{g^*_{ij}}$ (if the associated vertex $i$ has only one out-going edge $(i, j)$), or a sum of projection matrix $\sum_{j \in \mathcal{N}_i}P_{g^*_{ij}}$ (if the associated vertex $i$ has multiple out-going edges $(i, j)$), or is a zero block (if the associated vertex $i$ has no out-going edge).

According to Lemma~\ref{lemma:triangular_matrix}, the set of eigenvalues of a block triangular matrix is the union of eigenvalues of each diagonal block. In this bearing Laplacian,  each diagonal block is either a zero block or a positive semi-definite matrix (as a single projection or a sum of projection matrices). Therefore, according to Lemma~\ref{lemma:sum_projection}, the  eigenvalues of the
bearing Laplacian for directed acyclic graphs   are  real and  nonnegative. 
\end{proof}
\subsection{Conditions for bearing equivalence}
The following conditions are presented to characterize the null space property of $\LB$ for directed acyclic graphs.  Note that any directed acyclic graph contains a vertex with zero out-going edge, which is often termed as the ``leader agent" in formation control.  
\begin{proposition} \label{prop:non_equivalent_condition}
For a  bearing formation modeled by a directed acyclic graphs, any of the following conditions will result in  $\myspan\{\one\otimes I_d, p\} \subset  \Null(\LB)$, leading to a bearing non-equivalent formation.  
\begin{enumerate} \renewcommand{\theenumi}{\Roman{enumi}}
    \item There are at least two vertices that have zero out-going edge;
    \item There are at   least two vertices  with only one out-going edge;
    \item There are at least two vertices with collinear out-going edges. 
\end{enumerate}
 
\end{proposition}
\begin{proof}
Due to the block triangular structure of $\LB$ of directed acyclic graphs, any of the above conditions will introduce additional null spaces for $\LB$ in addition to $\myspan\{\one\otimes I_d, p\}$. The detailed proof is omitted here due to space limit. 
\end{proof} 

A special class of acyclic directed graphs is the \emph{leader-first-follower} (LFF) graph; see \cite{trinh2018bearing}.
\begin{definition}[{Leader-first-follower graph}]\label{definition_LFF}
A \emph{leader-first-follower} (LFF) graph is a directed graph on $n>1$ nodes such that
\begin{itemize}
\item[i)] One vertex (called the leader) has zero out-going edge.
    \item[ii)] One vertex (called the first follower) has one out-going edge and the corresponding edge is incident to the leader.
    \item[iii)] Every other vertex has at least two out-going edges  and their out-going edges are not collinear.
\end{itemize}
\end{definition}

We now give a sufficient condition for bearing equivalence, which characterizes the leader-first-follower (LFF) structure in directed graphs. 

\begin{theorem} \label{thm:LFF_GRAPH}
Directed formations over leader-first-follower graphs are bearing equivalent.
\end{theorem}
The proof of Theorem~\ref{thm:LFF_GRAPH} again follows from the  block triangular structure of $\LB$, and the detail is omitted here due to space limit. We also remark that this theorem can be seen as an extension of the distance persistence (\cite{hendrickx2007directed}, \cite{anderson2008rigid}) to bearing equivalence in directed LFF graphs. 

\section{Bearing equivalence in directed  graphs containing cycles} \label{Sec:BP_cyclic}


In this section, we focus on directed graphs containing cycles and derive several necessary and/or sufficient conditions on bearing Laplacian and bearing equivalence. 
\subsection{Spectrum of bearing Laplacian}
For directed graphs that contain cycles, the associated  bearing Laplacian (which is asymmetric) cannot be written in a block triangular structure, and thus its eigenvalues are often complex. 

The following conjecture on the bearing Laplacian spectrum was proposed in  \citep{zhao2015bearing}.  
\begin{conjecture}
The eigenvalues of the
bearing Laplacian $\LB$ of a directed formation have nonnegative
real parts. 
\end{conjecture}

This conjecture is not true. For counterexamples, see the formation graphs in Fig.~\ref{fig_example_equivalent_cyclic} which will be discussed later. Note that the entries of the bearing Laplacian (and thus its eigenvalue locations) depend on the configurations (agents' positions). For a given bearing Laplacian associated with a cyclic directed graph and under some special positions, eigenvalues with negative real parts   can occur. 

\subsection{Conditions for bearing equivalence}

We first generalize a sufficient   condition for characterizing $\Null(\RB)=\Null(\LB)$   in $\mathbb{R}^2$ (from \cite[Proposition 1]{zhao2015bearing} to bearing formations in general-dimensional spaces.  
\begin{proposition}
For a directed formation $\mathcal{G}(p)$ in $\mathbb{R}^d$ ($d \geq 2$), if each agent has at most two non-collinear out-going edges, then the
formation satisfies  $\Null(\RB)=\Null(\LB)$.
\end{proposition}
The proof follows from the proof of \cite[Proposition 1]{zhao2015bearing} and is omitted here. We note that this condition is not necessary.   For a counterexample, see the graph (c) in Fig.~\ref{fig_example_equivalent_cyclic}. In this example,  agent 2 has three out-going edges while the bearing formation of   Fig.~\ref{fig_example_equivalent_cyclic}(c) satisfies $\Null(\RB)=\Null(\LB)$ and  still is bearing equivalent.

Now we provide a necessary condition for bearing equivalence in directed graphs. 
\begin{proposition} \label{prop:necessary}
For a directed bearing formation $\mathcal{G}(p)$, if the associated bearing Laplacian satisfies $\Null(\LB)=\myspan\{\one\otimes I_d, p\}$, then the underlying directed graph   contains a directed spanning tree.  
\end{proposition}
\begin{proof}
If the underlying directed graph does not contain a spanning tree, then the Laplacian matrix $L =  J^T  H$ will have additional null vector besides $\{\one\}$, implying that the augmented Laplacian matrix $\bar L = L \otimes I_d = \bar J^T \bar H$ will have additional null space besides $\myspan\{\one\otimes I_d\}$, i.e., $\myspan\{\one\otimes I_d\}   \subset \Null(\bar L)$.  According to Lemma~\ref{lemma_nullAB}, since the block diagonal matrix $\mydiag{(P_{g_k})}$ is always singular that leads to the null space $\myspan\{p\}$ of $\LB$,  there  holds $$\Null(\bar L) = \Null(\bar J^\T \bar H) \subset \Null(\bar J^\T  \mydiag{(P_{g_k})}\bar{H}  ) = \Null(\LB),$$
which implies $\myspan\{\one\otimes I_d, p\}   \subset \Null(\LB)$.  Thus, to ensure that $\Null(\LB)=\myspan\{\one\otimes I_d, p\}$, the underlying directed graph must contain a spanning tree. 
\end{proof}
 
Proposition~\ref{prop:necessary} provides a necessary condition for bearing equivalence that  holds for both acyclic and cyclic directed graphs. In particular, Condition (1) in Proposition~\ref{prop:non_equivalent_condition} violates the spanning tree condition, and therefore any bearing formation with two leader agents (i.e., two vertices with zero out-degree) are not bearing equivalent. 
 
\subsection{Growing bearing equivalent formations}
A full characterization of bearing equivalence in directed graphs containing cycles still remains an open problem. We now present an alternative characterization to \textit{grow} bearing equivalent formations, to more number of agents  or to a higher-dimensional space. 
\begin{proposition} \label{prop:growing_BP}
Consider a bearing equivalent formation $\mathcal{G}(p)$ in $\mathbb{R}^d$ with $n$ agents modelled by a directed graph $\mathcal{G}$. Suppose an additional vertex (agent) $p'$ is added to $\mathcal{G}(p)$ with at least two out-going non-collinear edges  incident to existing vertices in $\mathcal{G}(p)$. Then this augmented formation $\mathcal{G}'(p, p')$ is bearing equivalent. In particular, if  $\myspan\{\one_n\otimes I_d, p\} = \Null(\LB (\mathcal{G}))$, then $\myspan\{\one_{n+1}\otimes I_d, (p, p')\} = \Null(\LB(\mathcal{G}'))$. 
\end{proposition}

\begin{proof} Without loss of generality we  
assign the index `$n+1$' to the new vertex $p'$ in the augmented graph $\mathcal{G}'$. The   augmented bearing Laplacian $\LB(\mathcal{G}')$ with the augmented graph $\mathcal{G}'$ can be expressed by
\begin{align} \label{eq:LB_new}
\LB(\mathcal{G}'(p, p')) = 
    \left[ \begin{array}{ *{4}{c} }
    & & &   0 \\
    & & & \vdots  \\
    \multicolumn{3}{c}
      {\raisebox{\dimexpr\normalbaselineskip+.7\ht\strutbox-.5\height}[0pt][0pt]
        {\scalebox{1.2}{$\LB (\mathcal{G}(p))$}}} & 0 \\
    \cdots & - P_{g_{(n+1)j}} & \cdots &  \sum_{j\in\N_{n+1}}P_{g_{(n+1)j}}   
  \end{array} \right]
\end{align}
The condition that the new vertex $p'$ has at least two out-going non-collinear edges connected with existing vertices guarantees that no additional null vector is introduced in $\LB(\mathcal{G}')$ with the edge addition. Therefore, if $\rank(\RB(\mathcal{G})) = \rank(\LB(\mathcal{G}))$, then it holds that $\rank(\RB(\mathcal{G}')) = \rank(\LB(\mathcal{G}'))$.
In particular, with the matrix structure in \eqref{eq:LB_new}, it is straightforward to show that if $\myspan\{\one_n\otimes I_d, p\} = \Null(\LB (\mathcal{G}))$, then $\myspan\{\one_{n+1}\otimes I_d, (p, p')\} = \Null(\LB(\mathcal{G}'))$. 
\end{proof}
This proposition holds for both acyclic and cyclic directed formations, and therefore can be used to analyze complex bearing formations if they can be decomposed by simple sub-graphs consisting of vertices with non-collinear  out-going edges. 
 
The following statement shows that a bearing equivalent formation in a lower-dimensional space will remain bearing equivalent in a higher-dimensional space.   
\begin{proposition} (Dimensional invariance) \label{prop: dim_invariance}
Bearing
equivalence is invariant to space dimensions. 
\end{proposition}
This proposition can be seen as a generation of the dimensional invariance property of infinitesimal bearing rigidity discovered in \cite{zhao2014TACBearing}.  
 

\subsection{Examples}
Fig.~\ref{fig_example_equivalent_cyclic} shows several examples of bearing equivalent formations modelled by directed graphs with cycles. As a consequence of the cyclic structure in these graphs, their bearing Laplacians have complex eigenvalues. For some special agents' positions, the associated bearing Laplacian can have eigenvalues with negative real parts.

Take the graph (a) in Fig.~\ref{fig_example_equivalent_cyclic} as an example. We randomly choose agents' positions and the bearing Laplacian of the graph (a) often delivers eigenvalues with negative real parts. One such position is 
\begin{align*}
p_1 &= [5.8009, 
    0.1698 ]^\T, p_2 =  [   1.2086,
    8.6271 ]^\T, \\
p_3 &= [
    4.8430
    8.4486 ]^\T,
    p_4 = [
    2.0941,
    5.5229]^\T .
    \end{align*}
The spectrum of the bearing Laplacian is computed as below 
\begin{align*}
\lambda_{1,2,3} &= 0, \lambda_{4,5} =   1.6400 \pm 0.7564i, \\
\lambda_{6,7} &=    0.8879 \pm 0.3799i, \lambda_{8} = -0.0559,
\end{align*}
which gives  only one  eigenvalue with negative real part. {\footnote{Matlab code for evaluating bearing equivalence in directed bearing formation examples is available at \url{https://www.dropbox.com/sh/5qzk9lqab0s448k/AACGEOr1r4riQVg-RqjmNDO4a?dl=0}. }}

For the example of graph (c) in Fig.~\ref{fig_example_equivalent_cyclic}, it can be constructed by adding agent 2 with three out-going edges to agents 1-3-4 in a  cyclic triangle formation  (which is bearing equivalent). Therefore, the formation of Fig.~\ref{fig_example_equivalent_cyclic}(c) is bearing equivalent   according to   Proposition~\ref{prop:growing_BP}. 
For all the bearing equivalent formations evaluated in the 2D space in Fig.~\ref{fig_example_equivalent_cyclic}, they remain bearing equivalent when agents' positions are lifted in the 3-D or higher-dimensional space according to the dimensional invariance property in Proposition~\ref{prop: dim_invariance}.

\begin{figure}[t]
\centering
\subfloat[]{
\centering
\resizebox{2.5cm}{!}{
\begin{tikzpicture}[
roundnode/.style={circle, draw=black, thick, minimum size=1.5mm,inner sep= 0.25mm},
squarednode/.style={rectangle, draw=red!60, fill=red!5, very thick, minimum size=5mm},
every node/.style={sloped,allow upside down}
]
\node[roundnode]   (v4)   at   (0,0) { };
\node[roundnode]   (v3)   at   (2,0) { };
\node[roundnode]   (v2)   at   (2,2) { };
\node[roundnode]   (v1)   at   (0,2) { };
\draw[-, very thick] (v2)--pic{arrow=latex}(v1);
\draw[-, very thick] (v1)--pic{arrow=latex}(v4);
\draw[-, very thick] (v4)--pic{arrow=latex}(v3);
\draw[-, very thick] (v3)--pic{arrow=latex}(v2);
\draw[-, very thick] (v4)--pic{arrow=latex}(v2);
\node (a1) at (-0.3,-0.3){$3$};
\node (a2) at (-0.3,2.3){$1$};
\node (a3) at (2.3,2.3){$2$};
\node (a4) at (2.3,-0.3){$4$};
\end{tikzpicture}
}
}
\hfill
\subfloat[]{
\centering
\resizebox{2.5cm}{!}{
\begin{tikzpicture}[
roundnode/.style={circle, draw=black, thick, minimum size=1.5mm,inner sep= 0.25mm},
squarednode/.style={rectangle, draw=red!60, fill=red!5, very thick, minimum size=5mm},
every node/.style={sloped,allow upside down}
]
\node[roundnode]   (v4)   at   (0,0) { };
\node[roundnode]   (v3)   at   (2,0) { };
\node[roundnode]   (v2)   at   (2,2) { };
\node[roundnode]   (v1)   at   (0,2) { };
\draw[-, very thick] (v2)--pic{arrow=latex}(v1);
\draw[-, very thick] (v1)--pic{arrow=latex}(v4);
\draw[-, very thick] (v4)--pic{arrow=latex}(v3);
\draw[-, very thick] (v2)--pic{arrow=latex}(v3);
\draw[-, very thick] (v4)--pic{arrow=latex}(v2);
\node (a1) at (-0.3,-0.3){$3$};
\node (a2) at (-0.3,2.3){$1$};
\node (a3) at (2.3,2.3){$2$};
\node (a4) at (2.3,-0.3){$4$};
\end{tikzpicture}
}
}
\hfill
\subfloat[]{
\centering
\resizebox{2.5cm}{!}{
\begin{tikzpicture}[
roundnode/.style={circle, draw=black, thick, minimum size=1.5mm,inner sep= 0.25mm},
squarednode/.style={rectangle, draw=red!60, fill=red!5, very thick, minimum size=5mm},
every node/.style={sloped,allow upside down}
]
\node[roundnode]   (v4)   at   (0,0) { };
\node[roundnode]   (v3)   at   (2,0) { };
\node[roundnode]   (v2)   at   (2,2) { };
\node[roundnode]   (v1)   at   (0,2) { };
\draw[-, very thick] (v2)--pic{arrow=latex}(v1);
\draw[-, very thick] (v1)--pic{arrow=latex}(v3);
\draw[-, very thick] (v2)--pic{arrow=latex}(v3);
\draw[-, very thick] (v2)--pic{arrow=latex}(v4);
\draw[-, very thick] (v3)--pic{arrow=latex}(v4);
\draw[-, very thick] (v4)--pic{arrow=latex}(v1);
\node (a1) at (-0.3,-0.3){$3$};
\node (a2) at (-0.3,2.3){$1$};
\node (a3) at (2.3,2.3){$2$};
\node (a4) at (2.3,-0.3){$4$};
\end{tikzpicture}
}
}
\caption{Examples of bearing \emph{equivalent} formations modelled by directed graphs with cycles.}
\label{fig_example_equivalent_cyclic}
\end{figure}
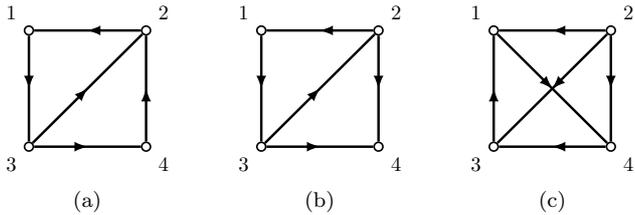

\section{Conclusion}\label{sec:con} \label{Sec:conclusion}
In this paper, motivated by the problem of bearing-based formation control in directed graphs, we present several conditions to characterize bearing equivalence for directed bearing formations. The notion of bearing equivalence is defined by the kernel equivalence of   bearing rigidity matrix and bearing Laplacian. These conditions for bearing equivalence are divided into two cases: bearing formations with acyclic directed graphs (that do not contain any cycle) and bearing formations with cyclic directed graphs (that contain at least one cycle). Several necessary and/or sufficient conditions are derived to reveal the properties of the spectrum and null space of the associated bearing Laplacian matrix. 

The notion of bearing equivalence emerges as one of the key properties to critically affect  stability and convergence of bearing-based formation systems modelled by directed graphs. In our future work, we will aim to present a complete characterization of bearing equivalence and apply the obtained conditions to solve the bearing-based formation control and network localization problems underpinned by directed graphs.

\section*{Acknowledgment}
The authors would like to thank Dr. Minh Hoang Trinh for helpful discussions and inputs to this work.

\bibliography{IFAC_bearing}             

\end{document}